\documentclass[journal]{IEEEtran} 

\usepackage[cmex10]{amsmath}
\usepackage{amsfonts,amssymb}
\usepackage{array}

\ifCLASSOPTIONcaptionsoff
  \usepackage[nomarkers]{endfloat}
 \let\MYoriglatexcaption\caption
 \renewcommand{\caption}[2][\relax]{\MYoriglatexcaption[#2]{#2}}
\fi

\usepackage{cite}
\usepackage{bbm} 
\usepackage[pdftitle={Coding on infinite alphabets},pdfauthor={Dominique Bontemps},pdfsubject={Information Theory,Statistics},pdfkeywords={Data compression, countable alphabets, redundancy, Bayes, adaptive compression.},colorlinks=true,linkcolor=black,urlcolor=black,citecolor=black,bookmarks=true,bookmarksnumbered=true]{hyperref}

\newcommand{\E}{\ensuremath{\mathbb{E}}}
\newcommand{\1}{\mathbbm{1}}

\newcommand{\N}{\ensuremath{\mathbb{N}}}
\newcommand{\R}{\ensuremath{\mathbb{R}}}
\newcommand{\Proba}{\ensuremath{\mathbb{P}}}

\newcommand{\m}{\ensuremath{\boldsymbol{\widetilde{m}}}}

\DeclareMathOperator{\vol}{Vol}

\newcommand{\ud}{\mathrm{d}}

\newtheorem{thm}{Theorem}
\newtheorem{lem}{Lemma}
\newtheorem{corol}{Corollary}
\newtheorem{prop}{Proposition}
\newtheorem{df}{Definition}

\begin{document}

\title{Universal Coding on Infinite Alphabets:\\ Exponentially Decreasing Envelopes}

\author{Dominique~Bontemps
\thanks{D. Bontemps is PhD fellowship holder at Laboratoire de Math\'ematiques, Univ. Paris-Sud 11, Orsay, 91405 France.}%
}

%


\maketitle

\begin{abstract}
This paper deals with the problem of universal lossless coding on a countable infinite alphabet. It focuses on some classes of sources defined by an envelope condition on the marginal distribution, namely exponentially decreasing envelope classes with exponent $\alpha$.

The minimax redundancy of exponentially decreasing envelope classes is proved to be equivalent to $\frac{1}{4 \alpha \log e} \log^2 n$. Then, an adaptive algorithm is proposed, whose maximum redundancy is equivalent to the minimax redundancy.
\end{abstract}


\IEEEpeerreviewmaketitle

\section{Introduction}

\IEEEPARstart{C}{ompression} of data is broadly used in our daily life: from the movies we watch to the office documents we produce. In this article, we are interested in lossless data compression on an unknown alphabet. This has applications in areas such as language modeling or lossless multimedia codecs.

First, we present briefly the problematics of data compression. More details are available in general textbooks, like \cite{eec_IEEE:cover_thomas_91}. Then we make a short review of preceding results, in which we situate the topic of this article, exponentially decreasing envelope classes, and we announce our results.

\subsection{Lossless data compression}

Consider a finite or countably infinite alphabet ${\mathcal X}$. A source on ${\mathcal X}$ is a probability distribution $\boldsymbol{P}$, on the set ${\mathcal X}^{\N}$ of infinite sequences of symbols from ${\mathcal X}$. Its marginal distributions are denoted by $P^n$, $n \geq 1$ (for $n=1$, we only note $P$). The scope of lossless data compression is to encode a sequence of symbols $X_{1:n}$, generated according to $P^n$, into a sequence of bits as small as possible. The algorithm has to be uniquely decodable.

The binary entropy $H(P^n) = \E_{P^n} [-\log_2 P^n(X_{1:n})]$ is known to be a lower bound for the expected codelength of $X_{1:n}$. From now on, $\log$ denotes the logarithm taken to base $2$, while $\ln$ is used to denote the natural logarithm. Since arithmetic coding based on $P^n$ encodes a message $x_{1:n}$ with $\lceil -\log P^n(x_{1:n})\rceil +1$ bits, this lower bound can be achieved within two bits. Then, the expected redundancy measures the mean number of extra bits, in addition to the entropy, a coding strategy uses to encode $X^n$. In the sequel, we use the word \emph{redundancy} instead of \emph{expected redundancy}.

Furthermore, together with Kraft-McMillan inequality, arithmetic coding provides an almost perfect correspondence between coding algorithms and probability distributions on ${\mathcal X}^{n}$. In this setting, if an algorithm is associated to the probability distribution $Q^n$, its expected redundancy reduces to the Kullbach-Liebler divergence between $P^n$ and $Q^n$ 
\[ D(P^n; Q^n) = \E_{P^n} \left[ \log \frac{P^n(X_{1:n})}{Q^n(X_{1:n})} \right]. \]
We call this quantity (expected) redundancy of the distribution $Q^n$ (with respect to $P^n$).

Unfortunately, the true statistics of the source are not known in general, but $P^n$ is supposed to belong to some large class of sources $\Lambda$ (for instance, the class of all iid sources, or the class of Markov sources). In this paper, the maximum redundancy
\[ R_n(Q^n; \Lambda) = \sup_{\boldsymbol{P} \in \Lambda} R_n(Q^n; P^n) \]
measures how well a coding probability $Q^n$ behave on an entire class $\Lambda$. With this point of view, the best coding probability is a \emph{minimax} coding probability, that achieves the \emph{minimax redundancy}
\[ R_n(\Lambda) = \inf_{Q^n} R_n(Q^n; \Lambda). \]

Another way to measure the ability of a class of sources to be efficiently encoded is the \emph{Bayes redundancy}
\[ R_{n,\mu}(\Lambda) = \inf_{Q^n} \int_{\Lambda} R_n(Q^n; P^n) \,\ud\mu(\boldsymbol{P}) \]
where $\mu$ is a prior distribution on $\Lambda$ endowed with the topology of weak convergence and the Borel $\sigma$-field. Only one coding strategy achieves the Bayes redundancy: the Bayes mixture
\begin{equation*}\label{eq:Mixture}
 M_{n,\mu}(x_{1:n}) = \int_{\Lambda} P^n(x_{1:n}) \,\ud\mu(\boldsymbol{P}).
\end{equation*}
When $\Lambda$ is a class of iid sources on the set ${\mathcal X}=\N_{*}=\N \backslash \{0\}$, there is a natural parametrization of $\Lambda$ by $P_{\boldsymbol{\theta}}(j) = \theta_j$, with $\boldsymbol{\theta} = (\theta_1, \theta_2, \ldots) \in \Theta_\Lambda$.  $\Theta_\Lambda$ is then a subset of
\[ \Theta = \left\{ \boldsymbol{\theta}=(\theta_1,\theta_2,\ldots) \in [0,1]^{\N}: \sum_{i\geq 1} \theta_i = 1 \right\} \]
and it is endowed with the topology of pointwise convergence. In this case we write $\mu$ as a prior on $\Theta_\Lambda$.

Minimax redundancy and Bayes redundancy are linked by an important relation \cite{eec_IEEE:Gallager68,eec_IEEE:DavissonLeonGarcia80}; it is written here in the context of iid sources on a finite or countably infinite alphabet, but Haussler \cite{eec_IEEE:haussler96general} has shown that it can be generalized for all classes of stationary ergodic processes on a complete separable metric space.
\begin{thm} \label{thm:MinimaxMaximin}
 Let $\Lambda$ be a class of iid sources, such that the parameter set $\Theta_{\Lambda}$ is a measurable subset of $\Theta$. Let $n\geq 1$. Then
 \[ R_n(\Lambda) = \sup_{\mu} R_{n,\mu}(\Lambda), \]
 where the supremum is taken over all (Borel) probability measures on $\Theta_{\Lambda}$.
\end{thm}
The quantity $\sup_{\mu} R_{n,\mu}(\Lambda)$ is called \emph{maximin redundancy}. A prior whose Bayes redundancy corresponds to the maximin redundancy is said to be maximin, or least favorable. \\
Theorem~\ref{thm:MinimaxMaximin} says that maximin redundancy and minimax redundancy are the same. It provides a tool to calculate the minimax redundancy.

Before speaking about known results, let us make mention of other two notions. \\
With an asymptotic point of view, a sequence of coding probabilities $(Q_n)_{n\geq 1}$ is said to be weakly universal if the per-symbol redundancy tends to $0$ on $\Lambda$: $\sup_{\boldsymbol{P} \in \Lambda} \lim_{n\rightarrow \infty} \frac{1}{n} D(P^n; Q^n) = 0$. \\
Instead of the expected redundancy, many authors 
consider individual sequences. In this case, the \emph{minimax regret}
\[R_n^{*}(\Lambda) = \inf_{Q^n} \sup_{\boldsymbol{P}\in\Lambda} \sup_{x_{1:n}\in {\mathcal X}^n} \log \frac{P^n(x_{1:n})}{Q^n(x_{1:n})} \]
plays the role that the minimax redundancy plays with the expected redundancy.

\subsection{Exponentially decreasing envelope classes}

In the case of a finite alphabet of size $k$, many classes of sources have been studied in the literature, for which estimates of the redundancy have been provided. In particular we have the class of all iid sources (see \cite{eec_IEEE:Krichevsky81,eec_IEEE:XieB97,eec_IEEE:XieB00,eec_IEEE:BarronRissanenYu98,eec_IEEE:Catoni2001,eec_IEEE:DrmotaS04}, and references therein), whose minimax redundancy is
\[\frac{k-1}{2} \log \frac{n}{2\pi e} + \log \frac{\Gamma(1/2)^k}{\Gamma(k/2)} + o(1).\]
This last class can be seen as a particular case of a $(k-1)$-dimensional class of iid sources on a (possibly) bigger alphabet, for which we have a similar result under certain conditions (see \cite{eec_IEEE:ClarkeB90,eec_IEEE:ClarkeB94,eec_IEEE:Barron98}). Similar results are still available for classes of Markov processes and finite memory tree sources on a finite alphabet (see \cite{eec_IEEE:Krichevsky81,eec_IEEE:Davisson83,eec_IEEE:willems95contexttree,eec_IEEE:Atteson99}), and for $k$-dimensional classes of even non-iid sources on an arbitrary alphabet (see \cite{eec_IEEE:Rissanen84}). \\
The results become less precise when one considers infinite dimensional classes on a finite alphabet. A typical example is the class of renewal processes, for which we do not have an equivalent of the expected redundancy, but we know that it is lower and upper bounded by a constant times $\sqrt{n}$ (see \cite{eec_IEEE:CsiszarShields96,eec_IEEE:FlajoletS02}). \\
Eventually, it is well known that the class of stationary ergodic sources on a finite alphabet is weakly universal (see \cite{eec_IEEE:cover_thomas_91}). However, Shields \cite{eec_IEEE:Shields93} showed that this class does not admit non-trivial universal redundancy rates.

In the case of a countably infinite alphabet, the situation is significantly different. Even the class of all iid sources is not weakly universal (see \cite{eec_IEEE:Kieffer78,eec_IEEE:GyorfiPM94}). Kieffer characterized weakly universal classes in \cite{eec_IEEE:Kieffer78} (see also \cite{eec_IEEE:GyorfiPM93,eec_IEEE:GyorfiPM94}):
\begin{prop}
 A class $\Lambda$ of stationary sources on $\N_{*}$ is weakly universal if and only if there exists a probability distribution $Q$ on $\N_{*}$ such that for every $\boldsymbol{P} \in \Lambda$, $D(P; Q) < \infty$.
\end{prop}

In the literature, we find two main ways to deal with infinite alphabets. The first one \cite{eec_IEEE:OrlitskySZ04,eec_IEEE:OrlitskyS04,eec_IEEE:JevticOS05,eec_IEEE:OrlitskySVZ06,eec_IEEE:ShamirC04,eec_IEEE:ShamirDCC04,eec_IEEE:GemelosW06,eec_IEEE:Garivier06Pattern,eec_IEEE:ChoiSzpankowski07} separates the message into two parts: a description of the symbols appearing in the message, and the \emph{pattern} they form. Then the compression of patterns is studied.

A second approach \cite{eec_IEEE:Elias75,eec_IEEE:GyorfiPM93,eec_IEEE:HeY05,eec_IEEE:FosterSW02,eec_IEEE:BouGarGas06} studies collections of sources satisfying Kieffer's condition, and proposes compression algorithms for these classes. A result from \cite{eec_IEEE:BouGarGas06} indicates us such a way:
\begin{prop}
 Let $\Lambda$ be a class of iid sources over $\N_{*}$. Let the envelope function $f$ be defined by $f(x) = \sup_{\boldsymbol{P} \in \Lambda} P(x)$. Then the minimax regret verifies
 \[ R_n^{*}(\Lambda) < \infty \Leftrightarrow \sum_{x\in \N_{*}} f(x) < \infty. \]
\end{prop}

It is therefore quite natural to consider classes of iid sources with envelope conditions on the marginal distribution. In this article we study specific classes of iid sources introduced by \cite{eec_IEEE:BouGarGas06}, and called \emph{exponentially decreasing envelope classes}.
\begin{df}\label{def:ExpEnvClass}
   Let $C$ and $\alpha$ be positive numbers satisfying $C > e^{2\alpha}$. The exponentially decreasing envelope class $\Lambda_{C e^{-\alpha\cdot}}$ is the class of sources defined by
   \begin{equation*}
    \begin{split}
     \Lambda_{C e^{-\alpha\cdot}} &= \{ \boldsymbol{P}: \forall k \geq 1,\ P(k) \leq C e^{-\alpha k} \\
      & \qquad \text{and } \boldsymbol{P}\ \text{is stationary and memoryless.} \}
    \end{split}
   \end{equation*}
\end{df}
The first condition addresses mainly the queue of the distribution of $X_1$; it means that great numbers must be rare enough. It does not mean that the distribution is geometrical: if $C$ is big enough, many other distributions are possible. Furthermore we will see that the exact value of $C$ does not change significantly the minimax redundancy, unlike $\alpha$.

Since in this paper we are going to only talk about exponentially decreasing envelope classes, we simplify the notations $R_n(Q^n; \Lambda_{C e^{-\alpha\cdot}})$, $R_n(\Lambda_{C e^{-\alpha\cdot}})$, and $R_{n,\mu}(\Lambda_{C e^{-\alpha\cdot}})$ into $R_n(Q^n; C,\alpha)$, $R_n(C,\alpha)$, and $R_{n,\mu}(C,\alpha)$ respectively. The subset of $\Theta$ corresponding to $\Lambda_{C e^{-\alpha\cdot}}$ is denoted by 
\begin{equation} \label{eq:ThetaCalpha}
 \begin{split}
  \Theta_{C,\alpha} &= \{ \boldsymbol{\theta}=(\theta_1,\theta_2,\ldots) \in [0,1]^{\N}: \\
   &\qquad \sum_{i\geq 1} \theta_i = 1\ \text{and } \forall i \geq 1,\ \theta_i \leq C e^{-\alpha i} \}.
 \end{split}
\end{equation}


We present two main results about these classes.

In Section~\ref{sec:MinimaxRedundancy} we calculate the minimax redundancy of exponentially decreasing envelope classes, and we find that it is equivalent to $\frac{1}{4 \alpha \log e} \log^2 n$ as $n$ tends to the infinity. This rate is interesting for two main reasons. Up to our knowledge, exponentially decreasing envelope classes are the first family of classes on an infinite alphabet for which an equivalent of the minimax redundancy is known. Then, even the rate is new: until now only rates in $\log n$ or $\sqrt{n}$ have been obtained.

Once the minimax redundancy of a class of sources is known, we are interested in finding a minimax coding algorithm. Section~\ref{sec:code} proposes a new adaptive coding algorithm, and we show that its maximum redundancy is equivalent to the minimax redundancy of exponentially decreasing envelope classes.

Eventually, the Appendix contains some proofs and some auxiliary results used in the main analysis.

\section{Minimax redundancy} \label{sec:MinimaxRedundancy}

In this section we state our main result. Theorem~\ref{thm:ExpClassRedundancy} below gives an equivalent of the minimax redundancy of exponentially decreasing envelope classes. To get it, we use a result due to Haussler and Opper \cite{eec_IEEE:HausslerOpper97}.

\begin{thm} \label{thm:ExpClassRedundancy}
 Let $C$ and $\alpha$ be positive numbers satisfying $C > e^{2 \alpha}$. The minimax redundancy of the exponentially decreasing envelope class $\Lambda_{C e^{-\alpha\cdot}}$ verifies
 \[ R_n(C,\alpha) \underset{n \rightarrow \infty}{\sim} \frac{1}{4 \alpha \log e} \log^2 n. \]
\end{thm}

Theorem~\ref{thm:ExpClassRedundancy} improves on a preceding result of \cite[Theorem~7]{eec_IEEE:BouGarGas06}. In that article the following bounds of the minimax redundancy of exponentially decreasing envelope classes are given:
\begin{equation*}
 \begin{split}
  &\frac{1}{8 \alpha \log e} \log^2 n\, (1+o(1)) \\
   &\qquad \leq R_n(C,\alpha) \\
   &\qquad \leq \frac{1}{2 \alpha \log e} \log^2 n + O(1).
 \end{split}
\end{equation*}

In subsection~\ref{sec:Haussler} we outline the work done in \cite{eec_IEEE:HausslerOpper97}, and then we use it in subsection~\ref{sec:EnvClassRedundancy} to prove Theorem~\ref{thm:ExpClassRedundancy}. 
Eventually, we discuss in subsection~\ref{sec:PolyEnvClasses} the adaptation of this method to other envelope classes.

\subsection{From the metric entropy to the minimax redundancy} \label{sec:Haussler}

To study the redundancy of a class of sources, \cite{eec_IEEE:HausslerOpper97} considers the Hellinger distance between the first marginal distributions of each source. Bounds on the minimax redundancy are provided in terms of the metric entropy of the set of the first marginal distributions, with respect to the Hellinger distance. As a consequence, that method can be applied only to iid sources. However it is very efficient in the case of exponentially decreasing envelope classes.

First, we need to define the Hellinger distance and the metric entropy. In the case of sources on a countably infinite alphabet, the Hellinger distance can be defined in the following way:
\begin{df}
 Let $P$ and $Q$ two probability distributions on $\N_{*}$. Then the Hellinger distance between $P$ and $Q$ is defined by
 \[ h(P,Q) = \sqrt{ \sum_{k\geq 1} \left( \sqrt{P(k)} - \sqrt{Q(k)} \right)^2 }. \]
\end{df}

A related metric can be defined on the parameter set $\Theta$:
\[ d(\boldsymbol{\theta},\boldsymbol{\theta}') = h(P_{\boldsymbol{\theta}},P_{\boldsymbol{\theta}'}) = \sqrt{ \sum_{k\geq 1} \left( \sqrt{\theta_k} - \sqrt{\theta'_k} \right)^2 }. \]

From a metric we can define the \emph{metric entropy}. We need to define first some numbers.
\begin{df}
 Let $S$ be a subset of $\Theta$, and $\epsilon$ be a positive number.
 \begin{enumerate}
  \item We denote by ${\mathcal D}_{\epsilon}(S,d)$ the cardinality of the smallest finite partition of $S$ with sets of diameter at most $\epsilon$, or we set ${\mathcal D}_{\epsilon}(S,d)=\infty$ if no such finite partition exists.
  \item The metric entropy of $(S,d)$ is defined by
   \[ {\mathcal H}_{\epsilon}(S,d) = \ln {\mathcal D}_{\epsilon}(S,d).\,\footnote{We follow \cite{eec_IEEE:HausslerOpper97} in this definition of the metric entropy. Several authors use a slightly different definition, based on the covering number or the packing number.} \]
  \item An \emph{$\epsilon$-cover} of $S$ is a subset $A\subset S$ such that, for all $x$ in $S$, there is an element $y$ of $A$ with $d(x,y)<\epsilon$. The \emph{covering number} ${\mathcal N}_{\epsilon}(S,d)$ is the cardinality of the smallest finite $\epsilon$-cover of $S$, or we define ${\mathcal N}_{\epsilon}(S,d)=\infty$ if no finite $\epsilon$-cover exists.
  \item An \emph{$\epsilon$-separated subset} of $S$ is a subset $A\subset S$ such that, for all distinct $x,\,y$ in $A$, $d(x,y)>\epsilon$. The \emph{packing number} ${\mathcal M}_{\epsilon}(S,d)$ is the cardinality of the largest finite $\epsilon$-separated subset of $S$, or we define ${\mathcal M}_{\epsilon}(S,d)=\infty$ if arbitrary large $\epsilon$-separated subsets exist.
 \end{enumerate}
\end{df}

The following lemma explains how these numbers are linked. It is a classical result that can be found for instance in \cite{eec_IEEE:VaartWellner96}.
\begin{lem}\label{lem:numbers}
 Let $S$ be a subset of $\Theta$. For all $\epsilon>0$,
 \[ {\mathcal M}_{2 \epsilon}(S,d) \leq {\mathcal D}_{2 \epsilon}(S,d) \leq {\mathcal N}_{\epsilon}(S,d) \leq {\mathcal M}_{\epsilon}(S,d). \]
\end{lem}
Lemma~\ref{lem:numbers} enables us to choose the most convenient number to calculate the metric entropy.

From the metric entropy one can define the notion of metric dimension, which generalizes the classical notion of dimension. But the metric entropy lets us know in some way how dense the elements are in a set, even infinite dimensional.

Another quantity that \cite{eec_IEEE:HausslerOpper97} uses is the \emph{minimax risk for the $(1+\lambda)$-affinity}
\[ R_\lambda(\Lambda) = \inf_{Q} \sup_{\boldsymbol{\theta}\in \Theta_{\lambda}} \sum_{k\geq 1} P_{\boldsymbol{\theta}}(k)^{1+\lambda} Q(k)^{-\lambda}, \]
defined for all $\lambda>0$. \\
More precisions about the $(1+\lambda)$-affinity are given in 
\cite{eec_IEEE:HausslerOpper97}%
. See also \cite{eec_IEEE:BontempsrapportM2} for a special regard payed to envelope classes.

In the case of an envelope class $\Lambda_f$ defined by an integrable envelope function $f$, it is easy to see that $R_\lambda<\infty$ for all $\lambda>0$. Indeed the choice 
\[ Q(k) = \frac{f(k)}{\sum_{l\geq 1} f(l)} \]
leads to the relation
\[ R_\lambda \leq \left( \sum_{k\geq 1} f(k) \right)^{\lambda}. \]

We can now write a slightly modified version\,\footnote{The separation of the upper and lower bounds have no effect on the proof given by Haussler and Opper. A complete justification is available in \cite{eec_IEEE:BontempsrapportM2}.} of Theorem~5 of \cite{eec_IEEE:HausslerOpper97} in the context of data compression on an infinite alphabet.
\begin{thm}\label{thm:HausslerThm5}
 Let $\Lambda$ be a class of iid sources on $\N_{*}$, such that the parameter set $\Theta_{\Lambda}$ is a measurable subset of $\Theta$. Assume that there exists $\lambda>0$ such that $R_\lambda<\infty$. 
 Let $h(x)$ be a continuous, non-decreasing function defined on the positive reals such that, for all $\gamma\geq 0$ and $C>0$,
 \begin{enumerate}
  \item \[ \lim_{x\rightarrow\infty} \frac{h\left(C x \left(h(x)\right)^\gamma\right)}{h(x)} = 1 \]
   and
  \item \[ \lim_{x\rightarrow\infty} \frac{h\left(C x (\ln x)^\gamma\right)}{h(x)} = 1. \]
 \end{enumerate}
 Then
 \begin{enumerate}
  \item If
   \[ {\mathcal H}_{\epsilon}(\Theta_{\Lambda},d) \underset{\epsilon \rightarrow 0}{\sim} h\left(\frac{1}{\epsilon}\right), \]
   then
   \[ R_n(\Lambda) \underset{n \rightarrow \infty}{\sim} (\log e)\, h(\sqrt{n}).\,\footnote{The $(\log e)$ factor comes from the use of the logarithm taken to base $2$, in the definition of $R_n$.} \]
  \item\label{thm5:lower} If, for some $\alpha>0$ and $c>0$,
   \[ \liminf_{\epsilon\rightarrow 0} \frac{{\mathcal H}_\epsilon(\Theta_{\Lambda},d)}{ (1/\epsilon)^{\alpha} h(1/\epsilon)} \geq c, \]
   then
   \[ \liminf_{n\rightarrow\infty} \frac{R_n(\Lambda)}{n^{\alpha/(\alpha+2)} \left[ h(n^{1/(\alpha+2)}) \right]^{2/(\alpha+2)}} > 0. \]
  \item\label{thm5:upper} If, for some $\alpha>0$ and $C>0$,
   \[ \limsup_{\epsilon\rightarrow 0} \frac{{\mathcal H}_\epsilon(\Theta_{\Lambda},d)}{ (1/\epsilon)^{\alpha} h(1/\epsilon)} \leq C, \]
   then
   \[ \limsup_{n\rightarrow\infty} \frac{R_n(\Lambda)}{(n \ln n)^{\alpha/(\alpha+2)} \left[ h(n^{1/(\alpha+2)}) \right]^{2/(\alpha+2)}} < \infty. \]
 \end{enumerate}
\end{thm}
The conditions concerning the function $h$ mean that $h$ cannot grow too fast. For instance, $h$ can grow like $C (\ln x)^\beta$, with $\beta\geq 0$.

The first case in the theorem is the one we use for exponentially decreasing envelope classes. In this case, the fast decreasing envelope produces a ``not too big'' metric entropy. Theorem~\ref{thm:HausslerThm5} gives us an equivalent of the minimax redundancy of the class of sources when $n$ goes to the infinity. This turns out very useful, as it improves a preceding result of \cite{eec_IEEE:BouGarGas06}. However it is only an asymptotic result, without any convergence speed.

The second and the third items correspond to bigger classes of sources. In these cases the result is a bit less interesting: it gives a speed for the growth of the redundancy, but without the associated constant factor. Furthermore there is a gap of $(\ln n)^{\alpha/(\alpha+2)}$ between the lower bound of point~\ref{thm5:lower} and the upper bound of point~\ref{thm5:upper}. However it allows us to retrieve more or less a result of \cite{eec_IEEE:BouGarGas06} for another type of envelope classes.

We develop now these applications.

\subsection{The minimax redundancy of exponentially decreasing envelope classes} \label{sec:EnvClassRedundancy}

Here we want to prove Theorem~\ref{thm:ExpClassRedundancy} by applying Theorem~\ref{thm:HausslerThm5}. Thus we have to calculate the metric entropy of exponentially decreasing envelope classes. This is done by Proposition~\ref{EnvClassLowerEntropy}:

\begin{prop} \label{EnvClassLowerEntropy}
 Let $C$ and $\alpha$ be positive numbers satisfying $C > e^{2 \alpha}$. The metric entropy of the parameters set $\Theta_{C,\alpha}$ verifies
 \[ {\mathcal H}_\epsilon(\Theta_{C,\alpha},d) = (1+o(1)) \frac{1}{\alpha} \ln^2 (1/\epsilon), \]
 where $o(1)$ is a function $g(\epsilon)$ such that $g(\epsilon)\rightarrow 0$ as $\epsilon\rightarrow 0$.
\end{prop}

\begin{proof}[Proof of Theorem~\ref{thm:ExpClassRedundancy}]
 Just apply Theorem~\ref{thm:HausslerThm5}, with $h(x) = \frac{1}{\alpha} \ln^2 (x)$, to get the result.
\end{proof}

\begin{proof}[Proof of Proposition~\ref{EnvClassLowerEntropy}] 
 The outlines of the proofs of the next lemmas can be found in Appendix~\ref{sec:EnvClassesMetrEnt}. The corollaries are simple applications and the proofs are skipped.
 
 We start with general considerations. Let $\Lambda_f$ be the envelope class defined by the integrable envelope function $f$. Let $\Theta_f$ be the corresponding parameter set
 \begin{equation*}
 \begin{split}
  \Theta_{f} &= \{ \boldsymbol{\theta}=(\theta_1,\theta_2,\ldots) \in [0,1]^{\N}: \\
   &\qquad \sum_{i\geq 1} \theta_i = 1\ \text{and } \forall i \geq 1,\ \theta_i \leq f(i) \}.
 \end{split}
\end{equation*}

 The function $\boldsymbol{\theta} \mapsto (\sqrt{\theta_1},\sqrt{\theta_2},\ldots)$ is an isometry between the metric space $(\Theta_f,d)$ and the subset $A_f \cap \{ \|x\| = 1 \}$ of $\ell^2$, equipped with the classical euclidean norm $\|\cdot\|$, where $A_f$ is defined by
 \begin{equation}
  A_f = \{ (x_k)_{k\in\N^*} \in \ell^2: \ \forall k\in\N^*, \ 0 \leq x_k \leq \sqrt{f(k)} \}.
 \end{equation}
 The metric entropy of $(\Theta_f,d)$ can be calculated in this space.
 
 Next we truncate some coordinates, to work in a finite dimensional space instead of $\ell^2$. Together with an adequate use of Lemma~\ref{lem:numbers}, this helps us to obtain upper and lower bounds for the metric entropy of $(\Theta_f,d)$. We start with the upper bound.
 
 \begin{lem} \label{lem:UpperEntropy}
  Let $\Lambda_f$ be the envelope class defined by the integrable envelope function $f$, and let $\epsilon$ be a positive number. Let $N_\epsilon$ denote the integer
  \begin{equation*}
   N_\epsilon = \inf \left\{n\geq 1: \ \sum_{k\geq n+1} f(k) \leq \frac{\epsilon^2}{16} \right\}.
  \end{equation*}
  For $U\in\R^N$ and $a>0$, let $B_{\R^{N}}(U, a)$ denote the ball in $\R^n$ with center $U$ and radius $a$. Then
  \[ {\mathcal H}_\epsilon(\Theta_f,d) \leq N_\epsilon \ln (1/\epsilon) + 3 N_\epsilon \ln 2 + A(N_\epsilon) + B(\epsilon), \]
  where 
  \[ A(N) = -\ln \vol \left( B_{\R^{N}}(0, 1) \right) = \ln \frac{\Gamma(\frac{N}{2}+1)}{\pi^{\frac{N}{2}}} \]
  and
  \[ B(\epsilon) = \sum_{k=1}^{N_\epsilon} \ln \left(\sqrt{f(k)} + \frac{\epsilon}{4}\right). \]
  Furthermore
  \[ A(N_\epsilon) \underset{\epsilon \rightarrow 0}{\sim} \frac{N_\epsilon}{2} \ln N_\epsilon. \]
 \end{lem}
 Note that
 \[ -N_\epsilon \ln (1/\epsilon) - 2 N_\epsilon \ln 2 \leq B(\epsilon) \leq \frac{\epsilon}{4} N_\epsilon. \]
 These bounds on $B(\epsilon)$ show that $B(\epsilon)$ tends to decrease the upper bound, while $A(N_\epsilon)$ contributes to its growth. If $\ln N_\epsilon$ behaves like $\ln (1/\epsilon)$ up to a constant factor, then the upper bound given in Lemma~\ref{lem:UpperEntropy} corresponds to a constant times $N_\epsilon \ln N_\epsilon$, and we are concerned with the point~\ref{thm5:upper} of Theorem~\ref{thm:HausslerThm5}.
 
 If we apply Lemma~\ref{lem:UpperEntropy} to the case of exponentially decreasing envelope classes, we obtain the upper bounds
 \begin{gather*}
  N_\epsilon \leq \frac{2}{\alpha} \ln (1/\epsilon) + \frac{1}{\alpha} \ln \frac{16 C}{1-e^{-\alpha}} \\
  B(\epsilon) \leq -\frac{\alpha}{4} N_\epsilon^2 + \left(\frac{\ln C}{2} + \frac{1}{\sqrt{1-e^{-\alpha}}}\right) N_\epsilon.
 \end{gather*}
 This leads to
 \begin{corol} \label{corol:EnvClassUpperEntropy}
  Let $C$ and $\alpha$ be positive numbers satisfying $C > e^{2\alpha}$. The metric entropy of the parameter set $\Theta_{C,\alpha}$ 
  defined by (\ref{eq:ThetaCalpha}) 
  verifies
  \[ {\mathcal H}_\epsilon(\Theta_{C,\alpha},d) \leq (1+o(1)) \frac{1}{\alpha} \ln^2 (1/\epsilon), \]
  where $o(1)$ is a function $g(\epsilon)$ such that $g(\epsilon)\rightarrow 0$ as $\epsilon\rightarrow 0$.
 \end{corol}
 
 Now we need to get a lower bound on the metric entropy. In this case too, we want to truncate some coordinates to bring ourselves to a smaller finite dimensional space. This time we truncate the first coordinates. Let us consider the number
 \begin{equation*}
  l_f = \min \{ l\geq 0: \sum_{k\geq l+1} f(k) \leq 1 \}.
 \end{equation*}
 
 \begin{lem} \label{lem:LowerEntropy}
  Let $\Lambda_f$ be the envelope class defined by an integrable envelope function $f$, which verifies
  \[ \sum_{k\geq 1} f(k) \geq 2. \]
  Let $\epsilon>0$ be a positive number, and let $m\geq 1$ be an integer. Then
  \begin{equation*} 
   {\mathcal H}_\epsilon(\Theta_f,d) \geq \frac{1}{2} \sum_{k=l_f+1}^{l_f+m} \ln f(k) + m \ln \left(\frac{1}{\epsilon}\right) + A(m),
  \end{equation*}
  where $A(m)$ is defined as in Lemma~\ref{lem:UpperEntropy}:
  \[ A(m) = -\ln \vol \left( B_{\R^{m}}(0, 1) \right) \underset{m \rightarrow \infty}{\sim} \frac{m}{2} \ln m. \]
 \end{lem}
 Note that exponentially decreasing envelopes verify the condition  $\sum_{k\geq 1} f(k) \geq 2$. Indeed the envelope of exponentially decreasing envelope classes is
 \[ f(k) = \min(1,C e^{-\alpha k}), \]
 and the condition $C > e^{2\alpha}$ entails that $f(1)= f(2) = 1$.
 
 From Lemma~\ref{lem:LowerEntropy} we can infer the following corollary, with the choice $m=\left\lfloor \frac{2}{\alpha} \ln\left(\frac{1}{\varepsilon}\right) \right\rfloor$.
 
 \begin{corol} \label{corol:EnvClassLowerEntropy}
  Let $C$ and $\alpha$ be positive numbers satisfying $C > e^{2 \alpha}$. The metric entropy of the parameters set $\Theta_{C,\alpha}$ verifies
  \[ {\mathcal H}_\epsilon(\Theta_{C,\alpha},d) \geq (1+o(1)) \frac{1}{\alpha} \ln^2 (1/\epsilon), \]
  where $o(1)$ is a function $g(\epsilon)$ such that $g(\epsilon)\rightarrow 0$ as $\epsilon\rightarrow 0$.
 \end{corol}
 Note that the bound is the same as in Corollary~\ref{corol:EnvClassLowerEntropy}. Therefore this concludes the proof of Proposition~\ref{EnvClassLowerEntropy}.
\end{proof}

\subsection{What about other envelope classes?} \label{sec:PolyEnvClasses}
In \cite{eec_IEEE:BouGarGas06} the redundancy of another type of envelope classes is also studied. The \emph{power-law envelope class} $\Lambda_{C \cdot^{-\alpha}}$ is defined, for $C>1$ and $\alpha>1$, by the envelope function $f_{\alpha,C}(x) = \min(1,\frac{C}{x^\alpha})$. The bounds obtained in \cite[Theorem~6]{eec_IEEE:BouGarGas06} are
\begin{equation} \label{eq:BGGClassesPoly}
 \begin{split}
  &A(\alpha) n^{1/\alpha} \log \lfloor C \zeta(\alpha) \rfloor \\
   &\qquad\leq \R_n(\Lambda_{C\cdot^{-\alpha}}) \\
   &\qquad\leq \left( \frac{2 C n}{\alpha-1} \right)^{1/\alpha} (\log n)^{1-1/\alpha} + O(1),
 \end{split}
\end{equation}
where
\[ A(\alpha) = \frac{1}{\alpha} \int_{1}^{\infty} \frac{1-e^{-1/(\zeta(\alpha) u)}}{u^{1-1/\alpha}} \ud u, \]
and $\zeta$ denotes the classical function $\zeta(\alpha) = \sum_{k\geq 1} \frac{1}{k^\alpha}$, for $\alpha>1$.

If one adapts the calculus made earlier to the power-law envelope classes, one can get the following upper and lower bounds: \\
There are two (calculable) constants $K_1,K_2>0$ such that, for all $\epsilon>0$,
\[ K_1 \left(\frac{1}{\epsilon}\right)^{\frac{2}{\alpha-1}} \leq {\mathcal H}_\epsilon \leq K_2 (1+o(1)) \left(\frac{1}{\epsilon}\right)^{\frac{2}{\alpha-1}} \ln \left(\frac{1}{\epsilon}\right). \]
Unfortunately this formula leaves a gap between the lower bound and the upper bound. The application of Theorem~\ref{thm:HausslerThm5} makes the gap worse. Indeed the polynomial part $\left(\frac{1}{\epsilon}\right)^{\frac{2}{\alpha-1}}$ of the metric entropy causes an additional gap of $\log^{1/\alpha} n$. In practice the bounds are the following: \\
There are two (unknown) constants $C,c>0$ such that, for all $n\geq 1$,
\begin{equation} \label{eq:PowerLawClasses}
c (1+o(1)) n^{1/\alpha} \leq R_n(\Lambda_{C \cdot^{-\alpha}}) \leq C (1+o(1)) n^{1/\alpha} \log n.
\end{equation}

These inequalities improve in no way the result of \cite{eec_IEEE:BouGarGas06}. May a better calculus of the metric entropy improve either their lower bound or their upper bound? Anyway the metric entropy of power-law envelope classes is ``too big'' to efficiently apply Theorem~\ref{thm:HausslerThm5}: it does not leave the hope for an equivalence, as for exponentially decreasing envelope classes. To summarize, the strategy based on the metric entropy and Theorem~\ref{thm:HausslerThm5} turns out efficient for ``small'' classes of sources.

\section{AutoCensuring Code} \label{sec:code}

This section presents a new algorithm called AutoCensuring Code (\textup{\texttt{ACcode}}). It is in fact a modification of the Censuring Code proposed by Boucheron, Garivier and Gassiat in \cite{eec_IEEE:BouGarGas06}. We keep the idea that big symbols are very few, and must be encoded differently, with an Elias code. Smaller symbols are encoded by arithmetic coding based on Krichevsky-Trofimov mixtures, which are known to be effective for finite alphabets. Our innovation is a data-driven cutoff $M_i = \sup_{1 \leq k \leq i} X_k$ used to encode $X_{i+1}$: with this choice we do not need to know the exact parameters of the exponentially decreasing envelope.

\textup{\texttt{ACcode}} is a prefix code on the set of all finite length messages, and it works on line. Its maximum redundancy on an exponentially decreasing envelope class $\Lambda_{Ce^{-\alpha\cdot}}$ is equivalent to the minimax redundancy of this class of sources. Furthermore \textup{\texttt{ACcode}} is adaptive, as the same algorithm verifies this property with all exponentially decreasing envelope classes. This is formulated in the following theorem, proved in Appendix~\ref{sec:proof-ACcode}. \\
 Let $\textup{\texttt{ACcode}}(x_{1:n})$ denote the binary string produced by \textup{\texttt{ACcode}} when it encodes the message $x_{1:n}$, and let $l(\cdot)$ denote the length of a string.

\begin{thm} \label{thm:ACcodeRedundancy}
 For any positive numbers $C$ and $\alpha$ satisfying $C > e^{2 \alpha}$,
 \[ \sup_{\boldsymbol{P} \in \Lambda_{C e^{-\alpha\cdot}}} \E_{P^n}[l(\textup{\texttt{ACcode}}(X_{1:n})) - H(P^n)] \underset{n \rightarrow \infty}{\sim} R_n(C, \alpha). \]
\end{thm}
The difference between the redundancy of \textup{\texttt{ACcode}} and the minimax redundancy is not necessarily bounded: there may exist codes whose redundancy is smaller than the redundancy of \textup{\texttt{ACcode}}, but with a benefit negligible with respect to $\log^2 n$.

Additionally, Theorem~\ref{thm:ACcodeRedundancy} enables us to retrieve the upper bound of the minimax redundancy obtained in the section~\ref{sec:MinimaxRedundancy}.

Let us now define \textup{\texttt{ACcode}}. Let $n\geq 1$ be some positive integer, and let $x_{1:n}=x_1 x_2\ldots x_n$ be a string from $\N_{*}^n$ to be encoded. We define the sequence of maxima
\[ m_0 = 0 \text{ and } m_i = \sup_{1 \leq k \leq i} x_k, \text{ for all } 1 \leq i \leq n. \]

The sequence $(m_i)_{1\leq i \leq n}$ is non-decreasing, piecewise constant. For $1\leq i \leq n$, let $n_{i}^0=\sum_{j=1}^{i} \1_{m_j>m_{j-1}}$ be the number of plateaus between $1$ and $i$. For $1\leq k \leq n^0_n$, let $\widetilde{m}_k$ be the $k^\text{th}$ new maximum: 
\begin{equation} \label{eq:mmtilde}
 m_i = \widetilde{m}_{n_i^0}.
\end{equation}

We define also $\widetilde{m}_0=0$. Let string $\m$ be the sequence $(\widetilde{m}_1 - \widetilde{m}_{0} +1), \ldots, (\widetilde{m}_{n_n^0} - \widetilde{m}_{n_n^0-1} +1), 1$. $\m$ is encoded into a binary string $\textup{\texttt{C2}}$ by applying Elias penultimate code (see \cite{eec_IEEE:Elias75}) to each number in $\m$. It is a prefix code which uses $l_E(x)$ bits to encode a positive integer $x$, with
\begin{equation} \label{eq:EliasLength}
 \begin{split}
   l_E(1) &= 1, \\
   l_E(x) &= 1 + \lfloor \log x \rfloor + 2 \left\lfloor \log \lfloor \log x \rfloor + 1 \right\rfloor \quad \text{if } x\geq 2.
 \end{split}
\end{equation}

Meanwhile the sequence of censored symbols is encoded using side information from $\m$. Consider the censored sequence $\widetilde{x}_{1:n}=\widetilde{x}_1 \widetilde{x}_2\ldots \widetilde{x}_n$ defined by
\[ \widetilde{x}_{i} = x_i \1_{x_i\leq m_{i-1}} = \left\{ \begin{array}{ll}
                                                 x_i &\text{if } x_i\leq m_{i-1}, \\
                                                 0 &\text{otherwise.}
                                                \end{array} \right. \]
All symbols greater than $m_{i-1}$ are encoded together: they are replaced by the extra symbol $0$, and this extra symbol is encoded instead. $0$ has a special use in our setting: it makes the decoder to know when $m_i$ changes, and that the new value has to be read in $\textup{\texttt{C2}}$. We add at the end of $\widetilde{x}_{1:n}$ an additional $0$, which acts as a termination signal together with the last $1$ in $\m$. This makes our code to be prefix on the set of all finite length messages (whatever $n$).

Therefore we produce the binary string $\textup{\texttt{C1}}$ by arithmetic coding of $\widetilde{x}_{1:n}0$. The conditional coding probabilities are defined by
\begin{align*}
  Q_{i+1}(\widetilde{X}_{i+1}=j | X_{1:i} = x_{1:i}) &= \frac{n_i^j + \frac{1}{2}}{i + \frac{m_i + 1}{2}} \quad\text{if}\quad 1 \leq j \leq m_i, \\
  Q_{i+1}(\widetilde{X}_{i+1} = 0 | X_{1:i} = x_{1:i}) &= \frac{1/2}{i + \frac{m_i + 1}{2}},
\end{align*}
where for $j\geq 1$ and $i\geq 0$, $n_i^j$ is the number of occurrences of symbol $j$ in $x_{1:i}$ (with convention $n_0^j = 0$ for all $j\geq 1$).

If $i\leq n-1$, the event $\{\widetilde{X}_{i+1} = 0\}$ is equal to $\{X_{i+1} > M_i\}$. If $x_{i+1} = j > m_i$, then $n_i^j=0$, and we still have
\[ Q_{i+1}(\widetilde{X}_{i+1} = 0 | X_{1:i} = x_{1:i}) = \frac{n_i^j + \frac{1}{2}}{i + \frac{m_i + 1}{2}}. \]

In the sequel we note the coding probability used to encode the entire string $\widetilde{x}_{1:n}0$ by
\[ Q^{n+1}(\widetilde{x}_{1:n}0) = Q_{n+1}(0 | x_{1:n}) \prod_{i=0}^{n-1} Q_{i+1}(\widetilde{x}_{i+1} | x_{1:i}). \]

A remark we can do is that the symbol $0$ is always considered as new: when $x_{i+1}>m_i$, we encode $0$ but we increment the counter $n^{x_{i+1}}_{i}$. (This choice has been made to simplify the calculation of the redundancy of \textup{\texttt{ACcode}}, but we suspect that changing this behavior could improve the performances.)

Now we have defined $\textup{\texttt{C1}}$ and $\textup{\texttt{C2}}$, we have to describe how they are transmitted. To keep our code on line, we overlap these two strings in the following way.

Arithmetic code needs a certain amount of bits, say $l_{i}$, to send the first $i$ symbol of $\widetilde{x}_{1:n}$. Unfortunately, $l_{i}$ depends on whether $i=n+1$ or not. In previous case $l_{n+1} = \lceil -\log Q^{n+1}(\widetilde{x}_{1:n}0) \rceil +1$, and in later one $l_{i}$ depends on the following symbols and has to be computed. \\
\textup{\texttt{ACcode}} begins with $\textup{\texttt{C2}}$, by the transmission of the Elias code of $\widetilde{m}_1+1$. Then the transmission of $\textup{\texttt{C1}}$ is initiated. Suppose that $\widetilde{x}_{i} = 0$ and $n^0_{i}=k$. As soon as $l_{i}$ bits of $\textup{\texttt{C1}}$ have been sent, 
the \textup{\texttt{ACcode}} algorithm sends the Elias code of $\widetilde{m}_k-\widetilde{m}_{k-1}+1$. Then $\textup{\texttt{C1}}$ is transmitted again, from the next bit.

To decode the $i^{\text{th}}$ symbol in $\textup{\texttt{C1}}$, the knowledge of the current maximum $m_{i-1}$ is needed; it is obtained from the beginning of the string $\textup{\texttt{C2}}$. The decoder also needs the counters $(n_{i-1}^j)_{j\geq 1}$, which can be computed from the first $i-1$ decoded symbols. As soon as $l_i$ bits of $\textup{\texttt{C1}}$ have been received, $\widetilde{x}_i$ can be decoded. When the decoder meets a $0$ at the $i^{\text{th}}$ position, he knows that the next bits are the Elias code of the next symbol in $\m$, and deduces $m_{i}$ via (\ref{eq:mmtilde}). Since the Elias code is prefix, the decoder knows when he receives $\textup{\texttt{C1}}$ again. Then the $(i+1)^{\text{th}}$ symbol can be processed.

Fig.~\ref{fig:ExACcode} shows an illustration of the transmission process. In this example, the initial message is $x_{1:4} = 5, 3, 2, 7$. Then the message encoded in $\textup{\texttt{C1}}$ is $\widetilde{x}_{1:4}0 = 0, 3, 2, 0, 0$. $13$ bits are needed to transmit the second $0$, and $15$ bits for the last one. In $\textup{\texttt{C2}}$ we transmit $\m = 6, 3, 1$.

\begin{figure*}[!t]
 \centering
  \begin{tabular}{c@{}c@{}c@{}c@{}c}
   $\underbrace{01110}$ & $\underbrace{0110011001101}$ & $\underbrace{0101}$ 
     & $\underbrace{00}$ & $\underbrace{1}$ \\
   $\downarrow$ & 
     \begin{tabular}{c}
      beginning \\
      of $\textup{\texttt{C1}}$
     \end{tabular} &
     $\downarrow$ &
     \begin{tabular}{c}
      last bits \\
      of $\textup{\texttt{C1}}$
     \end{tabular} &
     $\downarrow$ \\
   \begin{tabular}{c}
     Elias code \\
     of $\widetilde{m}_1+1$
   \end{tabular} &
     & \begin{tabular}{c}
        Elias code \\ 
        of $\widetilde{m}_2-\widetilde{m}_1+1$
       \end{tabular} & 
     & \begin{tabular}{c}
        Elias code \\ 
        of $1$
       \end{tabular}
  \end{tabular}
 \caption{Example of \textup{\texttt{ACcode}}}
 \label{fig:ExACcode}
\end{figure*}

In the previous example, exact calculations have been performed, but this is not sensible for a practical implementation of arithmetic coding. Some rule is needed to set the precision in calculus, and it must be used by both coder and decoder. To avoid a too big extra redundancy caused by approximations, precision can grow as $n$ grows. For instance, calculations can be made in memory with a further precision of $2\,\lceil \log i \rceil$ bits, in addition to the $\lceil -\log Q^{i}(\widetilde{x}_{1:i}) \rceil +1$ bits needed to encode $x_{1:i}$; this insures that the extra redundancy is bounded. 



\appendices

\section{Metric entropy of exponentially decreasing envelope classes} \label{sec:EnvClassesMetrEnt}

We give here the outlines of the lemmas we stated in subsection~\ref{sec:EnvClassRedundancy}.

\begin{proof}[Proof of Lemma~\ref{lem:UpperEntropy}]
 $N_\epsilon$ denotes the threshold from which we want to truncate the coordinates. If $y=(y_n)_{n\geq 1}$ is an element of $A_f$, its truncated version is $\widetilde{y} = (y_n \1_{n\leq N_\epsilon})_{n\geq 1}$. One can check that
 \begin{equation*}
   \|y - \widetilde{y}\| \leq \frac{\epsilon}{4}.
 \end{equation*}
 
 Suppose now that $S$ is an $\epsilon/4$-cover of $\{y\in A_f: \ \forall n\geq N_\epsilon, y_n=0\}$. Let $z$ denote an element of $A_f$. Then it exists some $y\in S$ such that $\|\widetilde{z} - y\| \leq \epsilon/4$. Thus $\|z - y\| \leq \epsilon/2$, and $S$ is an $\epsilon/2$-cover of $A_f$. This leads to
 \begin{equation*}
  \begin{split}
   {\mathcal D}_\epsilon(\Theta_f,d) &\leq {\mathcal M}_{\epsilon/4}\left( \prod_{1\leq k\leq N_\epsilon} \left[0,\sqrt{f(k)}\right], \|\cdot\|_{\R^{N_\epsilon}} \right) \\
    &\leq \frac{\vol \left( \prod_{1\leq k\leq N_\epsilon} \left[ -\frac{\epsilon}{8}, \sqrt{f(k)} + \frac{\epsilon}{8} \right] \right)}{\vol \left( B_{\R^{N_\epsilon}}(0, \frac{\epsilon}{8}) \right)} \\
  \end{split}
 \end{equation*}
 
 A first consequence of that calculus is that ${\mathcal D}_\epsilon(\Theta_f,d)$ is finite for all $\epsilon>0$. The first assertion of Lemma~\ref{lem:UpperEntropy} is then obtained by applying the logarithm function.
 
 The rest of Lemma~\ref{lem:UpperEntropy} follows from the Feller bounds, in their version proposed by \cite[ch. XII]{eec_IEEE:whittaker_watson}:
 \begin{equation} \label{eq:stirling}
  \Gamma(x) = \sqrt{2\pi}\, x^{x-1/2} e^{-x} e^{\frac{\beta}{12 x}}, \quad\text{with } \beta \in [0,1].
 \end{equation}
\end{proof}

\begin{proof}[Proof of Lemma~\ref{lem:LowerEntropy}]
 Let $m\geq 1$ be an integer. We project the set $A_f \cap \{ \|x\| = 1 \}$ over the $m$-dimensional space 
 \[E_m = \{0\}^{l_f} \times \R^m \times \{0\}^{\{k:k\geq l_f+m+1\}}\] 
 generated by the coordinates from $l_f+1$ to $l_f+m$. This leads to
 \begin{equation*}
  \begin{split}
   {\mathcal D}_\epsilon(\Theta_f,d) &\geq {\mathcal N}_\epsilon \left( \prod_{k=l_f+1}^{l_f+m} \left[0,\sqrt{f(k)}\right], \|\cdot\|_{\R^{m}} \right) \\
    &\geq \frac{\vol \left( \prod_{k=l_f+1}^{l_f+m} \left[0,\sqrt{f(k)}\right] \right)}{\vol \left( B_{\R^{m}}(0, \epsilon) \right)}.
  \end{split}
 \end{equation*}
 It only remains to apply the logarithm function.
\end{proof}

\section{Redundancy of \textup{\texttt{ACcode}}} \label{sec:proof-ACcode}

\subsection{Moments of $M_n$} \label{sec:ProofExpectedMn} 

We first need a lemma which contains several useful results about the moments of $M_n$.

\begin{lem} \label{lem:ExpectedMn}
 Let $C$ and $\alpha$ be positive numbers satisfying $C > e^{2 \alpha}$. 
 Then, for all $n\geq1$,
 \begin{enumerate}
  \item \label{lem:Mn1}
   \[ \sup_{\boldsymbol{P} \in \Lambda_{C e^{-\alpha\cdot}}} \E_P[M_n] \leq \frac{1}{\alpha} \left( \ln n + \ln \frac{C}{1-e^{-\alpha}} +1 \right). \] 
  \item \label{lem:Mn2} \[ \sup_{\boldsymbol{P} \in \Lambda_{C e^{-\alpha\cdot}}} \E_P \left[ M_n \1_{M_n > \frac{1}{\alpha} \ln \frac{C n^2}{1-e^{-\alpha}}} \right] = O\left(\frac{\ln n}{n}\right). \]
  \item \label{lem:Mn3} \[ \sup_{\boldsymbol{P} \in \Lambda_{C e^{-\alpha\cdot}}} \E_P[M_n \ln M_n] = o(\ln^2 n). \]
 \end{enumerate}
\end{lem}

\begin{proof}
 Let $F$ denote the distribution function associated with $P$. For $t\geq 0$, we have
 \begin{equation*}
  \begin{split}
   \Proba(X_1>t) 
    &= \sum_{k\geq \lfloor t\rfloor +1} P(k) \\
    &\leq \frac{C}{1-e^{-\alpha}} e^{-\alpha (\lfloor t\rfloor +1)} \\
    &\leq e^{-\alpha (t - \beta)},
  \end{split}
 \end{equation*}
 where $\beta = \frac{1}{\alpha} \ln \frac{C}{1-e^{-\alpha}}$. Therefore $F(t) \geq G(t)$ for all $t\in\R$, where
 \begin{equation*} 
  G(t) = \1_{t \geq \beta} \left(1-e^{-\alpha (t - \beta)} \right).
 \end{equation*}
 We can identify in $G$ the distribution function of the random variable $\beta + Y$, where $Y$ follows the exponential distribution with parameter $\alpha$.
 
 Let $U_1, \ldots, U_n$ be $n$ iid random variables following the uniform distribution on $[0,1]$. For $1\leq i\leq n$, let us define
 \begin{equation*}
  \begin{split}
   X'_i &= F^{-1}(U_i) \\
   Y_i &= G^{-1}(U_i)-\beta,
  \end{split}
 \end{equation*}
 where $F^{-1}$ and $G^{-1}$ denote the pseudo-inverses of $F$ and $G$:
 \[ \forall t \in [0,1], \quad F^{-1}(t) = \inf \{x\in\R: F(x) \geq t\}. \]
 
 Then the $n$-dimensional vector $X'_{1:n} = (X'_1,\ldots,X'_n)$ has the same distribution as $X_{1:n}$, and the maxima $M'_n = \sup_{1\leq i\leq n} X'_i$ and $M_n$ follow the same distribution. 
 
 On the other hand, the relation $F \geq G$ entails $X'_i \leq \beta+Y_i$, for all $1\leq i\leq n$. As the consequence, if $M''_n = \sup_{1\leq i\leq n} Y_i$ denotes the maximum of all $Y_i$, we have $M'_n \leq \beta+M''_n$. Since the random variables $Y_i$ are independent, the probability distribution of $M''_n$ is easy to calculate. Indeed for all $t>0$, 
 \begin{equation*}
  \begin{split}
   \Proba(M''_n \leq t) &= \Proba(\forall\ 1\leq i\leq n,\ Y_i \leq t) \\
    &= \left(1-e^{-\alpha t}\right)^n.
  \end{split}
 \end{equation*}
 We can write down the density function of $M''_n$: 
 \[ f(t) = \left\{ \begin{array}{ll} n\, \alpha\, e^{-\alpha t} (1-e^{-\alpha t})^{n-1} & \text{if } t>0, \\ 0 & \text{otherwise.} \end{array} \right. \]
 
 Now we look for an upper bound of $\E[M_n]$ by taking advantage of the knowledge of that distribution:
 \begin{equation*}
  \begin{split}
   \E[M_n] &= \E[M'_n] \\
    &\leq \E[\beta+M''_n] \\
    &= \beta+ \int_{0}^{\infty} t\, n\, \alpha\, e^{-\alpha t} (1-e^{-\alpha t})^{n-1}\,\ud t \\
    &= \beta+  \int_{0}^{\infty} \left(1-(1-e^{-\alpha t})^{n}\right)\,\ud t
  \end{split}
 \end{equation*}
 integrating by parts. Use now the change of variables
 \begin{equation*}
  \left\{ \begin{array}{l}
   u = 1 - e^{-\alpha t} \\
   t = \frac{-\ln (1-u)}{\alpha}
  \end{array} \right.
 \end{equation*}
 \begin{equation*}
  \begin{split}
   \E[M_n] &\leq \beta+ \frac{1}{\alpha} \int_{0}^{1} \frac{1-u^n}{1-u} \ud u \\
    &\leq \frac{1}{\alpha}\left( \ln n + 1 + \ln \frac{C}{1-e^{-\alpha}} \right).
  \end{split}
 \end{equation*}
 
 Since the upper bound does not depend on $P$, that achieves the proof of the point~\ref{lem:Mn1}. We can handle the point~\ref{lem:Mn2} in the same way. 
 For all $t>0$, we have 
 \begin{equation*}
  \begin{split}
   \E\left[M_n \1_{M_n>\beta+t}\right] &\leq \E\left[(\beta+M''_n) \1_{M''_n>t}\right] \\
    &\leq \int_{t}^{\infty} (\beta+u) n\, \alpha\, e^{-\alpha u}\,\ud u \\
    &= n e^{-\alpha t} \left(t+\frac{1}{\alpha}+\beta\right).
  \end{split}
 \end{equation*}
 With $t=\frac{2}{\alpha} \ln n$, we get the second point of Lemma~\ref{lem:ExpectedMn}.
 
 The third item is similar. Since the function $x\mapsto x\ln x$ is increasing on $[1,+\infty)$ and $1\leq M'_n \leq \beta+M''_n$, we have
 \begin{equation*}
  \begin{split}
   &\E\left[M_n \ln M_n\right] \\
    &\quad\leq \E\left[(\beta+M''_n) \ln (\beta+M''_n)\right] \\
    &\quad= \E\left[\1_{M''_n\leq \beta} (\beta+M''_n) \ln (\beta+M''_n)\right] \\
     &\quad\qquad + \E\left[\1_{M''_n> \beta} \1_{M''_n\leq \frac{2}{\alpha} \ln n} (\beta+M''_n) \ln (\beta+M''_n)\right] \\
     &\quad\qquad + \E\left[\1_{M''_n> \beta} \1_{M''_n> \frac{2}{\alpha} \ln n} (\beta+M''_n) \ln (\beta+M''_n)\right] \\
    &\quad\leq 2\beta \ln (2\beta) + \frac{4}{\alpha} (\ln n) \ln \left(\frac{4}{\alpha} \ln n\right)  \\
     &\quad\qquad + \E\left[2 M''_n \ln (2 M''_n) \1_{M''_n> \frac{2}{\alpha} \ln n}\right] \\
    &\quad\leq 2\beta \ln (2\beta) + \left(\frac{4}{\alpha} \ln \frac{4}{\alpha}\right) \ln n + \frac{4}{\alpha} (\ln n) (\ln \ln n) \\
     &\quad\qquad + \E\left[4 {M''_n}^2 \1_{M''_n> \frac{2}{\alpha} \ln n}\right].
  \end{split}
 \end{equation*}
 Let us define
 \[ \gamma(n) = 2\beta \ln (2\beta) + \left(\frac{4}{\alpha} \ln \frac{4}{\alpha}\right) \ln n + \frac{4}{\alpha} (\ln n) (\ln \ln n). \]
 Note that $\gamma(n) = o(\ln^2 n)$. Then
 \begin{equation*}
  \begin{split}
   \E\left[M_n \ln M_n\right]
    &\leq \gamma(n) + \int_{\frac{2}{\alpha} \ln n}^{\infty} 4 u^2 n\, \alpha\, e^{-\alpha u}\,\ud u \\
    &= \gamma(n) + \frac{4 n e^{-2 \ln n}}{\alpha^2} (4 \ln^2 n + 4 \ln n + 2).
  \end{split}
 \end{equation*}
 Taking the supremum over $P$, we get
 \begin{equation*}
  \begin{split}
   \sup_{\boldsymbol{P} \in \Lambda_{C e^{-\alpha\cdot}}} \E_P\left[M_n \ln M_n\right] &\leq \gamma(n) + \frac{16 \ln^2 n + 16 \ln n + 8}{\alpha^2\, n} \\
    &= o(\ln^2 n).
  \end{split}
 \end{equation*}
\end{proof}

\subsection{Contribution of \textup{\texttt{C1}}} \label{sec:ContribC1Lemmas}

\begin{prop} \label{prop:ContribC1}
 Let $C$ and $\alpha$ be positive numbers satisfying $C > e^{2 \alpha}$. Then
 \begin{equation*}
  \begin{split}
   &\sup_{\boldsymbol{P} \in \Lambda_{C e^{-\alpha\cdot}}} \E_{P^n} [ - \log  Q^n(\widetilde{X}_{1:n}) - H(P^n) ] \\
     &\qquad \leq (1+o(1)) \frac{1}{4 \alpha \log e} \log^2 n.
  \end{split}
 \end{equation*}
\end{prop}

\begin{proof}
  We give here the sketch of the proof, and we delay the proofs of (\ref{eq:C1A1}), (\ref{eq:C1A2}), (\ref{eq:C1A3}), and (\ref{eq:C1A4}).

  Here we deal with the quantity 
  \begin{equation*} 
  (\textup{A}) = \sup_{\boldsymbol{P} \in \Lambda_{C e^{-\alpha\cdot}}} \E_{P^n} [ - \log  Q^n(\widetilde{X}_{1:n}) - H(P^n)].
  \end{equation*}
  that corresponds to the contribution of $\textup{\texttt{C1}}$. As we saw in Section~\ref{sec:code}, the coding probability $Q^n$ is based on Krichevsky-Trofimov mixtures. For $k\geq 1$, let $KT_k$ denote the usual Krichevsky-Trofimov mixture on the alphabet $\{1, \ldots, k\}$, whose conditional probabilities are, for all $0\leq i \leq n-1$ and for all $1\leq j \leq k$,
  \begin{equation*} 
  KT_k(X_{i+1}=j | X_{1:i} = x_{1:i}) = \frac{n_i^j + \frac{1}{2}}{i + \frac{k}{2}}.
  \end{equation*}
  Let us choose $k=m_n+1$. In this case, there is a simple relation between $KT_{m_n+1}$ and $Q^n$. For any sequence of $n$ positive integers $x_{1:n}\in\N_{*}^n$,
  \begin{multline*}
  Q_{i+1}(\widetilde{X}_{i+1}=\widetilde{x}_{i+1} | X_{1:i} = x_{1:i}) \\
    = \frac{ 2 i + 1 + m_n}{2 i + 1 + m_i} KT_{m_n+1}(X_{i+1}=x_{i+1} | X_{1:i} = x_{1:i}).
  \end{multline*}
  As a consequence, we can link the redundancy of $Q^n$ to the redundancy of $KT_{m_n+1}$:
  \begin{equation*}
    \log  Q^n(\widetilde{X}_{1:n}) = \log KT_{M_n+1}(X_{1:n}) +\sum_{i=0}^{n-1} \log \frac{ 2 i + 1+ M_n}{2 i + 1 + M_i}
  \end{equation*}
  and therefore
  \begin{equation*} 
  \begin{split}
    (\textup{A}) &= \sup_{\boldsymbol{P} \in \Lambda_{C e^{-\alpha\cdot}}} \Bigg( \overset{(\textup{A}_1)}{\overbrace{ \E_{P^n} \big[ -\log KT_{M_n+1}(X_{1:n}) - H(P^n)\big] }} \\
    &\qquad - \overset{(\textup{A}_2)}{\overbrace{ \E_{P^n} \bigg[ \sum_{i=0}^{n-1} \log \frac{ 2 i +1 + M_n}{2 i + 1 + M_i} \bigg] }} \Bigg).
  \end{split}
  \end{equation*}
  Note that $(\textup{A}_2)$ corresponds to the gain in redundancy of $Q^n$ with respect to $KT_{M_n+1}$. It illustrates the benefit of taking $M_i$ instead of $M_n$ as cutoff to encode $X_{i+1}$.

  On the one hand, we have
  \begin{equation} \label{eq:C1A1}
    (\textup{A}_1) \leq \frac{\E[M_n]}{2} \log n + \E[\log (M_n+1)].
  \end{equation}
  Since $\E[\log M_n] \leq \E[M_n]$, Lemma~\ref{lem:ExpectedMn} entails
  \[\sup_{\boldsymbol{P} \in \Lambda_{C e^{-\alpha\cdot}}} \E[\log (M_n+1)] = o(\log^2 n).\]

  Applying Lemma~\ref{lem:ExpectedMn} again, we see that $(\textup{A}_1)$ produces a redundancy equivalent to $\frac{1}{2 \alpha \log e} \log^2 n$, which is twice bigger than the minimax redundancy obtained in Theorem~\ref{thm:ExpClassRedundancy}. So, we will hope the corrective term $(\textup{A}_2)$ to be about $\frac{1}{4 \alpha \log e} \log^2 n$. \\
  To deal with $(\textup{A}_2)$, we use the concavity of the $\log$ function, and we group the terms in the sum, $M_n$ by $M_n$. Let $m=\big\lfloor \frac{n-1}{M_n} \big\rfloor$ be the number of bundles. \\
  To simplify the expression, we also neglect few terms at the beginning of the sum. Let $(h_n)_{n\geq 1}$ be a non-decreasing sequence of positive integers, such that $h_n \rightarrow \infty$ as $n \rightarrow \infty$, and let us define $\lambda_n = 2 h_n \log \left( 1+ \frac{1}{2 h_n}\right)$. Then
  \begin{equation} \label{eq:C1A2}
    (\textup{A}_2) \geq \lambda_n \E_{P^n} \left[ \sum_{k=h_n +1}^{m} \frac{M_n-M_{k M_n}}{2 (k +1)} \right].
  \end{equation}
  It is easy to verify that the function $x \mapsto x \log \left(1+\frac{1}{x}\right)$ is non-decreasing, and tends to $\log e$ when $x$ tends to the infinity; therefore $(\lambda_n)$ tends to $\log e$. We can write now
  \begin{align*}
    (\textup{A}) &\leq \sup_{\boldsymbol{P} \in \Lambda_{C e^{-\alpha\cdot}}} \Bigg( \frac{\E[M_n]}{2} \log n \\
      &\qquad - \lambda_n \E \left[ \sum_{k=h_n +1}^{m} \frac{M_n-M_{k M_n}}{2 (k +1)} \right] \Bigg) + o(\log^2 n) \\
  \begin{split}
    &\leq \frac{1}{2} \overset{(\textup{A}_3)}{\overbrace{ \sup_{\boldsymbol{P} \in \Lambda_{C e^{-\alpha\cdot}}} \E_{P^n} \left[ M_n \log n - \lambda_n M_n \sum_{k=h_n +1}^{m} \frac{1}{k + 1} \right] }} \\
      &\qquad + \frac{\lambda_n}{2} \overset{(\textup{A}_4)}{\overbrace{ \sup_{\boldsymbol{P} \in \Lambda_{C e^{-\alpha\cdot}}} \E_{P^n} \left[ \sum_{k=h_n +1}^{m} \frac{M_{k M_n}}{k + 1} \right] }} + o(\log^2 n)
  \end{split}
  \end{align*}

  Let us choose $h_n = \max \{1,\lfloor \ln n - 2 \rfloor\}$. Then
  \begin{equation} \label{eq:C1A3}
    (\textup{A}_3) = o(\log^2 n)
  \end{equation}
  \begin{equation} \label{eq:C1A4}
    (\textup{A}_4) \leq \frac{\log^2 n}{2 \alpha \log^2 e} + o(\log^2 n).
  \end{equation}

  Therefore we have
  \begin{equation*}
  \begin{split}
    (\textup{A}) 
      &\leq (1+o(1)) \frac{1}{4 \alpha \log e} \log^2 n
  \end{split}
  \end{equation*}
  which concludes the proof of Proposition~\ref{prop:ContribC1}.
\end{proof}

\begin{proof}[Proof of (\ref{eq:C1A1})]
 Let 
 \[ \widehat{P}^n(x_{1:n}) = \sup_{P^n} P^n(x_{1:n}) = \prod_{j\in \{x_1,\ldots,x_n\}} \left(\frac{n_j^n}{n}\right)^{n_j^n} \]
 be the maximum likelihood of the string $x_{1:n}$ over all iid distribution on $\N^n$. Then
 \begin{equation*}
  \begin{split}
   (\textup{A}_1) &\leq \E_{P^n} \left[ \log \frac{\widehat{P}^n(X_{1:n})}{KT_{M_n+1}(X_{1:n})} \right] \\
    &\leq \E_{P^n} \left[ \sup_{x_{1:n} \in \{1, \ldots, M_n+1\}} \log \frac{\widehat{P}^n(x_{1:n})}{KT_{M_n+1}(x_{1:n})} \right]
  \end{split}
 \end{equation*}
 Now we can apply a result from Catoni \cite[prop~1.4.1]{eec_IEEE:Catoni2001}: 
 \begin{lem}
  For all $k\geq 1$ and for all $x_{1:n} \in \{1, \ldots, k\}^n$,
  \[ -\log KT_{k}(x_{1:n}) + \log \widehat{P}^n(x_{1:n}) \leq \frac{k-1}{2} \log n + \log k. \]
 \end{lem}
\end{proof}

\begin{proof}[Proof of (\ref{eq:C1A2})]
 We group the terms in $(\textup{A}_2)$, $M_n$ by $M_n$:
 \begin{equation*}
  \begin{split}
   (\textup{A}_2) 
    &\geq \E_{P^n} \left[ \sum_{k=1}^{m-1} \sum_{i=k M_n +1}^{(k+1) M_n} \log \left(1+ \frac{M_n-M_i}{2 i + M_i +1}\right) \right].
  \end{split}
 \end{equation*}
 
 From the relation $M_k \leq M_{k'}$ for all $k'\geq k \geq 1$, we can infer, for all $i\geq k M_n$,
 \[ \frac{M_n-M_i}{2 i + M_i +1} \leq \frac{M_n}{2 k M_n} = \frac{1}{2 k}  \]
 Since $\log$ is a concave function, we have $\log (1 + x) \geq \frac{x \log (1 + a)}{a}$ for all $a>0$ and $0\leq x\leq a$. Consequently, if we choose $a=\frac{1}{2 k}$,
 \begin{equation*}
  \begin{split}
   (\textup{A}_2) &\geq \E_{P^n} \left[ \sum_{k=1}^{m-1} \sum_{i=k M_n +1}^{(k+1) M_n} 2 k \log \left(1+ \frac{1}{2 k}\right) \frac{M_n-M_i}{2 i + M_i +1} \right] \\
    &\geq \E_{P^n} \left[ \sum_{k=h_n +1}^{m} \lambda_n \frac{M_n-M_{k M_n}}{2 k +2} \right].
  \end{split}
 \end{equation*}
\end{proof}

\begin{proof}[Proof of (\ref{eq:C1A3})]
 We have
 \begin{equation*}
  \begin{split}
   (\textup{A}_3) &= \sup_{\boldsymbol{P} \in \Lambda_{C e^{-\alpha\cdot}}} \Bigg[ \sum_{j\geq 1} P^n(M_n=j) \\
    &\qquad\qquad \times \Bigg( j \log n - \lambda_n j \sum_{k=h_n + 1}^{\left\lfloor \frac{n-1}{j} \right\rfloor} \frac{1}{k + 1} \Bigg) \Bigg]. 
  \end{split}
 \end{equation*}
 Then we plug in $h_n = \lfloor \ln n - 2 \rfloor$. For $n$ large enough, $h_n\geq 1$, and we have
 \begin{equation*}
  \begin{split}
   j \sum_{k=h_n + 1}^{\left\lfloor \frac{n-1}{j} \right\rfloor} \frac{1}{k + 1} &\geq j \int_{\ln n - 1}^{\left\lfloor \frac{n-1}{j} \right\rfloor + 1} \frac{\ud x}{x + 1} \\
   &= j\left( \ln \left(\left\lfloor\frac{n-1}{j}\right\rfloor + 2\right) - \ln (\ln n) \right) \\
   &\geq j \ln (n-1) - j \ln j - j \ln (\ln n),
  \end{split}
 \end{equation*}
 and therefore
 \begin{equation*}
  \begin{split}
   (\textup{A}_3) &\leq \sup_{\boldsymbol{P} \in \Lambda_{C e^{-\alpha\cdot}}} \big[(\log e - \lambda_n) \E[M_n] \ln n \\
    &\qquad\qquad + \lambda_n \E[M_n] \ln \frac{n}{n-1} \\
    &\qquad\qquad + \lambda_n \E[M_n \ln M_n] + \lambda_n \E[M_n] \ln (\ln n)\big].
  \end{split}
 \end{equation*}
 Then, if we use Lemma~\ref{lem:ExpectedMn} and the fact that $\lambda_n$ tends to $\log e$, we get (\ref{eq:C1A3}).
\end{proof}

\begin{proof}[Proof of (\ref{eq:C1A4})]
 We want to commute the expected value and the sum in $(\textup{A}_4)$. To do it, we need to get rid of $m$. We can note that the condition $k\leq m = \lfloor \frac{n-1}{M_n} \rfloor$ entails $k M_n \leq n-1$. Consequently, for $n$ big enough,
 \begin{equation*}
  \begin{split}
   (\textup{A}_4) &\leq \sup_{\boldsymbol{P} \in \Lambda_{C e^{-\alpha\cdot}}} \E_{P^n} \left[ \sum_{k=3}^{m} \frac{M_{k M_n}}{k + 1} \right] \\
    &\leq \sup_{\boldsymbol{P} \in \Lambda_{C e^{-\alpha\cdot}}} \E_{P^n} \left[ \sum_{k=3}^{n-1} \frac{M_{k M_n} \1_{k M_n \leq n-1}}{k + 1} \right] \\
    &\leq \sum_{k=3}^{n-1} \frac{\sup_{\boldsymbol{P} \in \Lambda_{C e^{-\alpha\cdot}}} \E_{P^n} [M_{k M_n} \1_{M_n \leq l_n}]}{k + 1} \\
     &\quad + \sup_{\boldsymbol{P} \in \Lambda_{C e^{-\alpha\cdot}}} \E_{P^n} [M_{n} \1_{M_n > l_n}] \sum_{k=3}^{n-1} \frac{1}{k + 1},
  \end{split}
 \end{equation*}
 where $l_n = \left\lfloor \frac{1}{\alpha} \left( 2\ln n + \ln \frac{C}{1-e^{-\alpha}}\right) \right\rfloor$. We can now plug in the results of Lemma~\ref{lem:ExpectedMn}:
 \begin{equation*}
  \begin{split}
   (\textup{A}_4) &\leq \sum_{k=3}^{n-1} \frac{\ln (k l_n) + 1 + \ln \frac{C}{1-e^{-\alpha}}}{(k + 1)\alpha} + o(1) \sum_{k=3}^{n-1} \frac{1}{k + 1} \\
    &\leq \frac{1}{\alpha} \sum_{k=3}^{n-1} \frac{\ln k}{k + 1} + \frac{1}{\alpha} (\ln l_n + O(1)) \sum_{k=3}^{n-1} \frac{1}{k + 1}.
  \end{split}
 \end{equation*}
 Note that $l_n = O(\ln n)$, and consequently $\ln l_n = O(\ln \ln n)$. So
 \begin{equation*}
  \begin{split}
   (\textup{A}_4) &\leq \frac{1}{\alpha} \int_{3}^{n} \frac{\ln x}{x}\,\ud x + O(\ln \ln n) \int_{3}^{n} \frac{\ud x}{x} \\
    &\leq \frac{1}{2 \alpha} \ln^2 n + o(\ln^2 n).
  \end{split}
 \end{equation*}
\end{proof}

\subsection{Contribution of \textup{\texttt{C2}}} \label{sec:ContribC2Lemmas}

\begin{prop} \label{prop:ContribC2}
 Let $C$ and $\alpha$ be positive numbers satisfying $C > e^{2 \alpha}$. Then
 \[ \sup_{\boldsymbol{P} \in \Lambda_{C e^{-\alpha\cdot}}} \E_{P^n} [l(\textup{\texttt{C2}})] \leq o(\log^2 n). \]
\end{prop}

\begin{proof}
 Like in the previous subsection, we give first the sketch of the proof, and we delay several technical lemmas.
 \begin{multline*}
  \sup_{\boldsymbol{P} \in \Lambda_{C e^{-\alpha\cdot}}} \E_{P^n} [l(\textup{\texttt{C2}})] \\
  \leq 1 + \sup_{\boldsymbol{P} \in \Lambda_{C e^{-\alpha\cdot}}} \sum_{i=1}^{n} \E_{P^n} \left[  \1_{X_i>M_{i-1}} l_E(X_i+1) \right].
 \end{multline*}

  We deal with this sum thanks to the following lemma:
  \begin{lem} \label{lem:C2point1}
  Let $g$ be a positive and non-decreasing function on $[1,\infty)$. Let $(K_n)_{n\geq 1}$ be a non-decreasing sequence of positive integers. Then, for all $n\geq 1$,
  \begin{multline*}
    \sup_{\boldsymbol{P} \in \Lambda_{C e^{-\alpha\cdot}}} \sum_{i=1}^{n} \E_{P^n} \left[  \1_{X_i>M_{i-1}} g(X_i) \right]  \\
    \leq (1+o(1)) g(K_n) \frac{1}{\alpha} \ln n + C\, n \int_{K_n}^{\infty} g(x+1) e^{-\alpha x} \,\ud x.
  \end{multline*}
  \end{lem}

  To apply Lemma~\ref{lem:C2point1}, we extend the definition of $l_E$ on $[1,\infty)$ by
  \[ l_E(x) = \left\{ \begin{array}{ll} 1 & \text{if } x \in [1,2), \\ 1 + \lfloor \log x \rfloor + 2 \left\lfloor \log \lfloor \log x \rfloor + 1 \right\rfloor & \text{if } x\geq 2. \end{array} \right. \]
  We get
 \begin{multline*}
  \sup_{\boldsymbol{P} \in \Lambda_{C e^{-\alpha\cdot}}} \E_{P^n} [l(\textup{\texttt{C2}})] \\
  \leq (1+o(1)) \frac{l_E(K_n+1)}{\alpha} \ln n + C\, n \int_{K_n}^{\infty} l_E(x+2) e^{-\alpha x} \,\ud x
 \end{multline*}

  Then we can choose $K_n=\max \{1,\lfloor \frac{1}{\alpha} \ln n \rfloor\}$. This entails
  \[l_E(K_n+1) \underset{n \rightarrow \infty}{\sim} \log K_n \sim \log \log n = o(\log n),\] 
  and therefore
  \[ \frac{1}{\alpha} l_E(K_n+1) \ln n = o(\log^2 n). \]

  The remaining term is treated by Lemma~\ref{lem:C2point2}, which achieves the proof of Proposition~\ref{prop:ContribC2}:

  \begin{lem} \label{lem:C2point2}
  Let $\alpha>0$ be a real number, and let $K_n=\max \{1,\lfloor \frac{1}{\alpha} \ln n \rfloor\}$. Then
  \[ n \int_{K_n}^{\infty} l_E(x+2) e^{-\alpha x} \,\ud x = o(\log n). \]
  \end{lem}
\end{proof}

\begin{proof}[Proof of Lemma~\ref{lem:C2point1}]
 Let $\boldsymbol{P}$ be an element of $\Lambda_{C e^{-\alpha\cdot}}$. Let us define, for all $k\geq 0$,
 \[ \bar{p}(k) = P(X_1>k) = \sum_{j\geq k+1} P(k), \] 
 and
 \[ (\textup{B}_1) = \sum_{i=1}^{n} \E_{P^n} \left[  \1_{X_i>M_{i-1}} g(X_i) \right]. \]
 Note that, for all $1\leq i\leq n$, $X_i$ and $M_{i-1}$ are independent random variables, and
 \begin{equation*}
  \begin{split}
   P^n(M_i \leq k) &= P^n(\forall\ 1\leq j \leq i,\ X_j \leq k) \\ 
    &= (1-\bar{p}(k))^i.
  \end{split}
 \end{equation*}
 Then we can write
 \begin{equation*}
  \begin{split}
   (\textup{B}_1) &= \sum_{i=1}^{n} \sum_{k\geq 0} P^n(M_{i-1} = k) \sum_{m\geq k+1} P(m) g(m) \\ 
    &= \sum_{m\geq 1} P(m) g(m) \sum_{i=1}^{n} \sum_{k=0}^{m-1} \Proba(M_{i-1} = k) \\
    &= P(1) g(1) + \sum_{m\geq 2} P(m) g(m) \sum_{i=1}^{n} (1-\bar{p}(m-1))^{i-1} \\
    &= \sum_{m\geq 1} P(m) g(m) \frac{1-(1-\bar{p}(m-1))^{n}}{\bar{p}(m-1)}.
  \end{split}
 \end{equation*}
 If we take $g(x)=1$ for all $x$, we get
 \begin{equation*}
  \begin{split}
   \sum_{m\geq 1} P(m) \frac{1-(1-\bar{p}(m-1))^{n}}{\bar{p}(m-1)} &= \E \left[ \sum_{i=1}^{n} \1_{X_i>M_{i-1}} \right] \\
    &\leq \E [M_n]. 
  \end{split}
 \end{equation*}
 In the general case, we can split the sum at $K_n$, and we get
 \begin{equation*}
  \begin{split}
   (\textup{B}_1) &= \sum_{m=1}^{K_n} P(m) g(m) \frac{1-(1-\bar{p}(m-1))^{n}}{\bar{p}(m-1)} \\
     &\qquad + \sum_{m\geq K_n+1} P(m) g(m) \frac{1-(1-\bar{p}(m-1))^{n}}{\bar{p}(m-1)} \\
    &\leq g(K_n) \sum_{m\geq 1} P(m) \frac{1-(1-\bar{p}(m-1))^{n}}{\bar{p}(m-1)} \\
     &\qquad + \sum_{m\geq K_n+1} n P(m) g(m) \\
    &\leq g(K_n) \E[ M_n] + C\, n \sum_{m\geq K_n+1} g(m) e^{-\alpha m}.
  \end{split}
 \end{equation*}
 At this point, we can take the supremum over all sources $\boldsymbol{P}$ in $\Lambda_{C e^{-\alpha\cdot}}$:
 \begin{equation*}
  \begin{split}
   &\sup_{\boldsymbol{P} \in \Lambda_{C e^{-\alpha\cdot}}} \sum_{i=1}^{n} \E_{P^n} \left[  \1_{X_i>M_{i-1}} g(X_i) \right]  \\
   &\quad\leq (1+o(1)) g(K_n) \frac{1}{\alpha} \ln n + C\, n \int_{K_n}^{\infty} g(x+1) e^{-\alpha x} \,\ud x.
  \end{split}
 \end{equation*}
\end{proof}

\begin{proof}[Proof of Lemma~\ref{lem:C2point2}]
 \begin{equation*}
  \begin{split}
   &n \int_{K_n}^{\infty} l_E(x+2) e^{-\alpha x} \,\ud x \\
    &\quad\leq n \int_{K_n}^{\infty} \left( \log (x+2) +2 \log \log (x+3) + 1 \right) e^{-\alpha x} \,\ud x \\
    &\quad\leq n e^{-\alpha K_n} \log (K_n+3) \\
     &\quad\qquad\int_{K_n +3}^{\infty} \frac{\log x +2 \log \log x + 1}{\log (K_n+3)} e^{-\alpha (x-K_n -3)} \,\ud x \\
    &\quad\leq e^{\alpha} \log (K_n+3) \left(\sup_{x\geq K_n+3} \frac{\log x +2 \log \log x + 1}{\log x}\right) \\ 
     &\quad\qquad \int_{K_n +3}^{\infty} \frac{\log x}{\log (K_n+3)} e^{-\alpha (x-K_n -3)} \,\ud x \\
    &\quad= O(\log K_n) \int_{0}^{\infty} \left(1+\frac{\log \left(1+\frac{x}{K_n+3}\right)}{\log (K_n+3)}\right) e^{-\alpha x} \,\ud x \\
    &\quad= o(\log n).
  \end{split}
 \end{equation*}
 The supremum is correctly defined and bounded, because the function
 \[ x \mapsto \frac{\log x +2 \log \log x + 1}{\log x} \]
 is continuous and tends to $1$ as $x$ tends to the infinity.
\end{proof}

\subsection{Proof of Theorem~\ref{thm:ACcodeRedundancy}} 

The message sent by the \textup{\texttt{ACcode}} algorithm is compound of two strings $\textup{\texttt{C1}}$ and $\textup{\texttt{C2}}$. $\textup{\texttt{C1}}$ corresponds to the part of the message encoded by the arithmetic code, with coding probability $Q^{n+1}$. The arithmetic code encodes a message $\widetilde{x}_{1:n}0$ with $\lceil - \log Q^{n+1}(\widetilde{x}_{1:n}0) \rceil +1$ bits. We have
 \begin{equation*}
 \begin{split}
  \E_{P^n} [- \log Q_{n+1}(0 | X_{1:n})] 
   &= \E_{P^n} [\log (M_n +1 + 2n) ] \\
   &\leq \log (2n) + \frac{\E_{P^n} [ M_n+1 ]}{2n} \\
   &= O(\log n)
 \end{split}
 \end{equation*}
 thanks to Lemma~\ref{lem:ExpectedMn}. Therefore the redundancy of \textup{\texttt{ACcode}} can be upper bounded, for all $n\geq 2$, by
 \begin{multline*}
 \sup_{\boldsymbol{P} \in \Lambda_{C e^{-\alpha\cdot}}} \E_{P^n} [ l(\textup{\texttt{C1}})+l(\textup{\texttt{C1}})] - H(P^n)]  \\
  \begin{aligned}
   &\leq\sup_{\boldsymbol{P} \in \Lambda_{C e^{-\alpha\cdot}}} \E_{P^n} [ - \log  Q^n(\widetilde{X}_{1:n}) - H(P^n) ] \\
   &\quad + \sup_{\boldsymbol{P} \in \Lambda_{C e^{-\alpha\cdot}}} \E_{P^n} [l(\textup{\texttt{C2}})] + O(\log n).
  \end{aligned}
 \end{multline*}
 We conclude thanks to Propositions~\ref{prop:ContribC1} and \ref{prop:ContribC2}.

\section*{Acknowledgment}

I wish thank Elisabeth Gassiat for her helpful advice, for the many ideas in this paper she suggested me, and for her constant availability to my questions. 

Thanks also to Aur\'elien Garivier and to St\'ephane Boucheron for the useful discussions we had. I don't forget the ideas I got from my classmates, especially Rafik Imekraz. 

\ifCLASSOPTIONcaptionsoff
  \newpage
\fi


\begin{thebibliography}{10}
\providecommand{\url}[1]{#1}
\csname url@samestyle\endcsname
\providecommand{\newblock}{\relax}
\providecommand{\bibinfo}[2]{#2}
\providecommand{\BIBentrySTDinterwordspacing}{\spaceskip=0pt\relax}
\providecommand{\BIBentryALTinterwordstretchfactor}{4}
\providecommand{\BIBentryALTinterwordspacing}{\spaceskip=\fontdimen2\font plus
\BIBentryALTinterwordstretchfactor\fontdimen3\font minus
  \fontdimen4\font\relax}
\providecommand{\BIBforeignlanguage}[2]{{%
\expandafter\ifx\csname l@#1\endcsname\relax
\typeout{** WARNING: IEEEtran.bst: No hyphenation pattern has been}%
\typeout{** loaded for the language `#1'. Using the pattern for}%
\typeout{** the default language instead.}%
\else
\language=\csname l@#1\endcsname
\fi
#2}}
\providecommand{\BIBdecl}{\relax}
\BIBdecl

\bibitem{eec_IEEE:cover_thomas_91}
T.~M. Cover and J.~A. Thomas, \emph{Elements of Information Theory}.\hskip 1em
  plus 0.5em minus 0.4em\relax New York: Wiley, 1991.

\bibitem{eec_IEEE:Gallager68}
R.~G. Gallager, \emph{Information Theory and Reliable Communication}.\hskip 1em
  plus 0.5em minus 0.4em\relax New York: Wiley, 1968.

\bibitem{eec_IEEE:DavissonLeonGarcia80}
L.~D. Davisson and A.~Leon-Garcia, ``A source matching approach to finding
  minimax codes,'' \emph{{IEEE} Trans. Inf. Theory}, vol.~26, pp. 166--174,
  1980.

\bibitem{eec_IEEE:haussler96general}
D.~Haussler, ``A general minimax result for relative entropy,'' University of
  California, UC Santa Cruz, CA 96064, Tech. Rep. UCSC-CRL-96-26, 1996.

\bibitem{eec_IEEE:Krichevsky81}
R.~E. Krichevsky and V.~K. Trofimov, ``The performance of universal encoding,''
  \emph{{IEEE} Trans. Inf. Theory}, vol.~27, no.~2, pp. 199--207, 1981.

\bibitem{eec_IEEE:XieB97}
Q.~Xie and A.~R. Barron, ``Minimax redundancy for the class of memoryless
  sources,'' \emph{{IEEE} Trans. Inf. Theory}, vol.~43, no.~2, pp. 646--657,
  1997.

\bibitem{eec_IEEE:XieB00}
------, ``Asymptotic minimax regret for data compression, gambling, and
  prediction,'' \emph{{IEEE} Trans. Inf. Theory}, vol.~46, no.~2, pp. 431--445,
  2000.

\bibitem{eec_IEEE:BarronRissanenYu98}
A.~R. Barron, J.~Rissanen, and B.~Yu, ``The minimum description length
  principle in coding and modeling,'' \emph{{IEEE} Trans. Inf. Theory},
  vol.~44, no.~6, pp. 2743--2760, 1998.

\bibitem{eec_IEEE:Catoni2001}
O.~Catoni, \emph{Statistical Learning Theory and Stochastic Optimization}, ser.
  Lecture Notes in Mathematics.\hskip 1em plus 0.5em minus 0.4em\relax
  Springer-Verlag, 2001, vol. 1851, {\'E}cole d'{\'E}t{\'e} de Probabilit{\'e}s
  de Saint-Flour XXXI.

\bibitem{eec_IEEE:DrmotaS04}
M.~Drmota and W.~Szpankowski, ``Precise minimax redundancy and regret,''
  \emph{{IEEE} Trans. Inf. Theory}, vol.~50, no.~11, pp. 2686--2707, 2004.

\bibitem{eec_IEEE:ClarkeB90}
B.~S. Clarke and A.~R. Barron, ``Information-theoretic asymptotics of bayes
  methods,'' \emph{{IEEE} Trans. Inf. Theory}, vol.~36, no.~3, pp. 453--471,
  1990.

\bibitem{eec_IEEE:ClarkeB94}
------, ``Jeffrey's prior is asymptotically least favorable under entropy
  risk,'' \emph{J. Statist. Plann. Inference}, vol.~41, pp. 37--60, 1994.

\bibitem{eec_IEEE:Barron98}
A.~R. Barron, ``Information-theoretic characterization of bayes performance and
  the choice of priors in parametric and nonparametric problems,'' in
  \emph{Bayesian Statistics}, J.~M. Bernardo, J.~O. Berger, D.~A. P., and
  S.~A.~F. M., Eds.\hskip 1em plus 0.5em minus 0.4em\relax Oxford Univ. Press,
  1998, vol.~6, pp. 27--52.

\bibitem{eec_IEEE:Davisson83}
L.~D. Davisson, ``Minimax noiseless universal coding for markov sources,''
  \emph{{IEEE} Trans. Inf. Theory}, vol.~29, no.~2, pp. 211--214, 1983.

\bibitem{eec_IEEE:willems95contexttree}
F.~M.~J. Willems, Y.~M. Shtarkov, and T.~J. Tjalkens, ``The context-tree
  weighting method: Basic properties,'' \emph{{IEEE} Trans. Inf. Theory},
  vol.~41, no.~3, pp. 653--664, 1995.

\bibitem{eec_IEEE:Atteson99}
K.~Atteson, ``The asymptotic redundancy of bayes rules for markov chains,''
  \emph{{IEEE} Trans. Inf. Theory}, vol.~45, no.~6, pp. 2104--2109, 1999.

\bibitem{eec_IEEE:Rissanen84}
J.~Rissanen, ``Universal coding, information, prediction, and estimation,''
  \emph{{IEEE} Trans. Inf. Theory}, vol.~30, no.~4, pp. 629--636, 1984.

\bibitem{eec_IEEE:CsiszarShields96}
I.~Csisz{\'a}r and P.~C. Shields, ``Redundancy rates for renewal and other
  processes,'' \emph{{IEEE} Trans. Inf. Theory}, vol.~42, no.~6, pp.
  2065--2072, 1996.

\bibitem{eec_IEEE:FlajoletS02}
P.~Flajolet and W.~Szpankowski, ``Analytic variations on redundancy rates of
  renewal processes,'' \emph{{IEEE} Trans. Inf. Theory}, vol.~48, no.~11, pp.
  2911--2921, 2002.

\bibitem{eec_IEEE:Shields93}
P.~C. Shields, ``Universal redundancy rates do not exist,'' \emph{{IEEE} Trans.
  Inf. Theory}, vol.~39, no.~2, pp. 520--524, 1993.

\bibitem{eec_IEEE:Kieffer78}
J.~C. Kieffer, ``A unified approach to weak universal source coding,''
  \emph{{IEEE} Trans. Inf. Theory}, vol.~24, no.~6, pp. 674--682, 1978.

\bibitem{eec_IEEE:GyorfiPM94}
L.~Gy{\"o}rfi, I.~P{\'a}li, and E.~C. van~der Meulen, ``There is no universal
  source code for an infinite source alphabet,'' \emph{{IEEE} Trans. Inf.
  Theory}, vol.~40, no.~1, pp. 267--271, 1994.

\bibitem{eec_IEEE:GyorfiPM93}
------, ``On universal noiseless source coding for infinite source alphabets,''
  \emph{Eur. Trans. Telecom.}, vol.~4, pp. 4--16, 1993.

\bibitem{eec_IEEE:OrlitskySZ04}
A.~Orlitsky, N.~P. Santhanam, and J.~Zhang, ``Universal compression of
  memoryless sources over unknown alphabets,'' \emph{{IEEE} Trans. Inf.
  Theory}, vol.~50, no.~7, pp. 1469--1481, 2004.

\bibitem{eec_IEEE:OrlitskyS04}
------, ``Speaking of infinity [i.i.d. strings],'' \emph{{IEEE} Trans. Inf.
  Theory}, vol.~50, no.~10, pp. 2215--2230, 2004.

\bibitem{eec_IEEE:JevticOS05}
N.~Jevtic, A.~Orlitsky, and N.~P. Santhanam, ``A lower bound on compression of
  unknown alphabets,'' \emph{Theor. Comput. Sci.}, vol. 332, no. 1-3, pp.
  293--311, 2005.

\bibitem{eec_IEEE:OrlitskySVZ06}
A.~Orlitsky, N.~P. Santhanam, K.~Viswanathan, and J.~Zhang, ``Limit results on
  pattern entropy,'' \emph{{IEEE} Trans. Inf. Theory}, vol.~52, no.~7, pp.
  2954--2964, 2006.

\bibitem{eec_IEEE:ShamirC04}
G.~I. Shamir and D.~J.~J. Costello, ``On the entropy rate of pattern
  processes,'' \emph{{IEEE} Trans. Inf. Theory}, vol.~50, no.~8, pp.
  1620--1635, 2004.

\bibitem{eec_IEEE:ShamirDCC04}
G.~I. Shamir, ``Sequential universal lossless techniques for compression of
  patterns and their description length,'' in \emph{Data Compression
  Conference}, 2004, pp. 419--428.

\bibitem{eec_IEEE:GemelosW06}
G.~M. Gemelos and T.~Weissman, ``On the entropy rate of pattern processes,''
  \emph{{IEEE} Trans. Inf. Theory}, vol.~52, no.~9, pp. 3994--4007, 2006.

\bibitem{eec_IEEE:Garivier06Pattern}
A.~Garivier, ``A lower bound for the maximin redundancy in pattern coding,''
  \emph{Entropy}, vol.~11, no.~4, pp. 634--642, 2009.

\bibitem{eec_IEEE:ChoiSzpankowski07}
Y.~Choi and W.~Szpankowski, ``Pattern matching in constrained sequences,'' in
  \emph{2007 Int. Symp. Information Theory}, Nice, 2007, pp. 2606--2610.

\bibitem{eec_IEEE:Elias75}
P.~Elias, ``Universal codeword sets and representations of the integers,''
  \emph{{IEEE} Trans. Inf. Theory}, vol.~21, no.~2, pp. 194--203, 1975.

\bibitem{eec_IEEE:HeY05}
D.~ke~He and E.~hui Yang, ``The universality of grammar-based codes for sources
  with countably infinite alphabets,'' \emph{{IEEE} Trans. Inf. Theory},
  vol.~51, no.~11, pp. 3753--3765, 2005.

\bibitem{eec_IEEE:FosterSW02}
D.~P. Foster, R.~A. Stine, and A.~J. Wyner, ``Universal codes for finite
  sequences of integers drawn from a monotone distribution,'' \emph{{IEEE}
  Trans. Inf. Theory}, vol.~48, no.~6, pp. 1713--1720, 2002.

\bibitem{eec_IEEE:BouGarGas06}
S.~Boucheron, A.~Garivier, and E.~Gassiat, ``Coding on countably infinite
  alphabets,'' \emph{{IEEE} Trans. Inf. Theory}, vol.~55, no.~1, pp. 358--373,
  2009.

\bibitem{eec_IEEE:HausslerOpper97}
D.~Haussler and M.~Opper, ``Mutual information, metric entropy and cumulative
  relative entropy risk,'' \emph{Ann. Statist.}, vol.~25, no.~6, pp.
  2451--1492, 1997.

\bibitem{eec_IEEE:VaartWellner96}
A.~W. van~der Vaart and J.~A. Wellner, \emph{Weak Convergence and Empirical
  Processes}, ser. Springer Series in Statistics.\hskip 1em plus 0.5em minus
  0.4em\relax Springer-Verlag, 1996.

\bibitem{eec_IEEE:BontempsrapportM2}
\BIBentryALTinterwordspacing
D.~Bontemps, ``\BIBforeignlanguage{french}{Redondance bay{\'e}sienne et
  minimax, sources stationnaires sans m{\'e}moire en alphabet infini},''
  Master's thesis, Paris-Sud 11 Univ., sep 2007. [Online]. Available:
  \url{http://www.math.u-psud.fr/~bontemps/Documents/rapportM2.pdf}
\BIBentrySTDinterwordspacing

\bibitem{eec_IEEE:whittaker_watson}
E.~T. Whittaker and G.~N. Watson, \emph{A Course of Modern Analysis}, 4th~ed.,
  ser. Cambridge Mathematical Library.\hskip 1em plus 0.5em minus 0.4em\relax
  Cambridge: Cambridge Univ. Press, 1927, reprinted 1990.

\end{thebibliography}

\end{document}